\documentclass{article}
\usepackage{arxiv}

\usepackage{hyperref}
\usepackage{microtype}

\usepackage{amssymb}
\usepackage{amsmath}
\usepackage{amsthm}
\usepackage{mathdots}
\usepackage{mathtools}

\usepackage{tensor}

\usepackage{urwchancal}
\DeclareFontFamily{OT1}{pzc}{}
\DeclareFontShape{OT1}{pzc}{m}{it}{<-> s * [1.10] pzcmi7t}{}
\DeclareMathAlphabet{\mathpzc}{OT1}{pzc}{m}{it}

\usepackage{graphicx}
\usepackage{ifpdf}
\ifpdf
\usepackage{epstopdf}
\PrependGraphicsExtensions{.svg}
\fi

\usepackage{xypic}

\newtheorem{theorem}{Theorem}[section]
\newtheorem{lemma}[theorem]{Lemma}

\newtheorem{proposition}[theorem]{Proposition}

\newtheorem{remark}[theorem]{Remark}

\usepackage{url}

\providecommand{\R}{\mathbb{R}}

\providecommand{\SO}{\mathbf{SO}}
\providecommand{\SL}{\mathbf{SL}}

\providecommand{\SE}{\mathbf{SE}}

\providecommand{\grpG}{\mathbf{G}}

\providecommand{\gothg}{\mathfrak{g}}

\providecommand{\se}{\mathfrak{se}}

\providecommand{\calN}{\mathcal{N}}

\providecommand{\vecV}{\mathbb{V}}

\providecommand{\Id}{I} 

\DeclareMathOperator{\tr}{tr}
\providecommand{\trace}[1]{\tr\left(#1\right)}

\DeclareMathOperator{\Ad}{Ad}

\providecommand{\pr}{\mathbb{P}} 

\providecommand{\Lyap}{\mathcal{L}} 

\providecommand{\tT}{\mathrm{T}} 
\providecommand{\td}{\mathrm{d}}
\providecommand{\tD}{\mathrm{D}}

\providecommand{\mr}[1]{{#1}^\circ} 
\providecommand{\ub}[1]{\underline{#1}}
\providecommand{\ob}[1]{\overline{#1}} 
\usepackage{accents}
\makeatletter
\providecommand{\scirc}{%
    \hbox{\fontfamily{\rmdefault}\fontsize{0.4\dimexpr(\f@size pt)}{0}\selectfont{\raisebox{-0.52ex}[0ex][-0.52ex]{$\circ$}}}}

\makeatother

\mathchardef\mhyphen="2D

\providecommand{\etal}{\textit{et al.}~}

\newcommand{\publicationdetails}
{
	\copyrightNoticeIEEE{2020}
	Ng, Y., Van Goor, P., Hamel, T., and Mahony, R. (2020), ``Equivariant Systems Theory and Observer Design for Second Order Kinematic Systems on Matrix Lie Groups'', \textit{2020 IEEE 59th Conference on Decision and Control (CDC)}, December 2020, pp. 4194–4199,
	\DOILink{https://ieeexplore.ieee.org/abstract/document/9303761}{10.1109/CDC42340.2020.9303761}
}
\newcommand{\publicationversion}
{Author accepted version}

\begin{document}

\title{Equivariant Systems Theory and Observer Design for Second Order Kinematic Systems on Matrix Lie Groups}
\headertitle{Equiv. Sys. Theory and Observer Design for $2^{nd}$ Order Kin. Sys. on Matrix Lie Groups}

\author{
\href{https://orcid.org/0000-0002-7764-298X}{\includegraphics[scale=0.06]{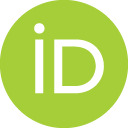}\hspace{1mm}
Yonhon Ng}
\\
Systems Theory and Robotics Group \\
Australian National University \\
ACT, 2601, Australia \\
\texttt{yonhon.ng@anu.edu.au} \\
\And	
\href{https://orcid.org/0000-0003-4391-7014}{\includegraphics[scale=0.06]{orcid.png}\hspace{1mm}
Pieter van Goor}
\\
Systems Theory and Robotics Group \\
Australian National University \\
ACT, 2601, Australia \\
\texttt{pieter.vangoor@anu.edu.au} \\
\And	
\href{https://orcid.org/0000-0002-7779-1264}{\includegraphics[scale=0.06]{orcid.png}\hspace{1mm}
Tarek Hamel}
\\
I3S (University C\^ote d'Azur, CNRS, Sophia Antipolis)\\
and Insitut Universitaire de France\\
\texttt{thamel@i3s.unice.fr} \\
\And	
\href{https://orcid.org/0000-0002-7803-2868}{\includegraphics[scale=0.06]{orcid.png}\hspace{1mm}
Robert Mahony}
\\
Systems Theory and Robotics Group \\
Australian National University \\
ACT, 2601, Australia \\
\texttt{robert.mahony@anu.edu.au} \\
}

\maketitle

\begin{abstract}
This paper presents the equivariant systems theory and observer design for second order kinematic systems on matrix Lie groups.
The state of a second order kinematic system on a matrix Lie group is naturally posed on the tangent bundle of the group with the inputs lying in the tangent of the tangent bundle known as the double tangent bundle.
We provide a simple parameterization of both the tangent bundle state space and the input space (the fiber space of the double tangent bundle) and then introduce a semi-direct product group and group actions onto both the state and input spaces.
We show that with the proposed group actions the second order kinematics are equivariant.
An equivariant lift of the kinematics onto the symmetry group is defined and used to design a nonlinear observer on the lifted state space using nonlinear constructive design techniques.
A simple hovercraft simulation verifies the performance of our observer.
\end{abstract}


\section{Introduction}

Observer design for second order kinematic systems is of interest since there are many applications where direct velocity measurements are not available but where acceleration is measured by an Inertial Measurement Unit (IMU) or can be inferred from torque or force measurements.
However, the majority of observer designs for systems on Lie groups and homogeneous spaces are developed for first order kinematics.
Attitude estimation is naturally posed on the Special Orthogonal group $\SO(3)$ \cite{Bonnabel08,Mahony08,Grip11,Hua14} where an IMU directly provides the angular velocity as a measured input.
Pose estimation on $\SE(3)$ has been considered \cite{Bras09,Baldwin09,Hua18} in the same framework, as has estimating homographies on the special linear group $\SL(3)$ \cite{Hamel11,Mahony12} although the inputs are not directly available and depend on additional parameter estimation.
The Simultaneous Localisation and Mapping (SLAM) problem has also been studied \cite{2016_Barrau_arxive,2017_Barrau_tac,Mahony17} in the same context.

The fact that an IMU measures angular velocity (the first order kinematic input of rotation) and acceleration (the second order kinematic input for translation) has motivated a number of works that use part first order and part second order kinematic models.
If linear velocity and linear acceleration are measured, either expressed in an inertial frame~\cite{Bonnable09}, body-fixed frame~\cite{Hua16} or a combination of both~\cite{Hua17}, velocity aided attitude estimation can be used to better estimate the attitude and linear velocity of an accelerated body.
Nonlinear observers on Special Euclidean group $\SE(3)$ where velocity, acceleration and landmark measurements are used to achieve local exponential stability result~\cite{Vasconcelos10,Hua11} have also been studied.
Br\'as \etal \cite{Bras09} estimates the pose of an unmanned aerial vehicle with acoustic ranging and inertial sensor and Hua \etal \cite{Hua18} developed a Riccati observer with landmark and inertial measurements.
However, these work do not exploit the underlying symmetry of the kinematic system in their observer design.
For pure second order kinematics, the authors earlier work considers the case of second order kinematics on $\SO(3)$~\cite{Ng19} and $\SE(3)$~\cite{Ng20a}.

In this paper we consider second order kinematic systems on general matrix Lie groups.
The second order kinematics can be written in left-trivialised form as first order kinematics on the Lie group coupled to velocity kinematics on the Lie-algebra.
We define a symmetry group comprised of a semi-direct product of the system group acting on its own Lie algebra and show that this group acts transitively on the left trivialised representation of the tangent bundle of the system state space.
A key contribution of the paper is to introduce a velocity input space that acts in both the first and second order components of the system kinematics modelling the fibre space of the double tangent bundle of the matrix Lie group of the system.
The component of the input velocity that acts in the first order kinematics is virtual (unless there are direct velocity measurements available), however, it is a core construction to understand the equivariance of the system.
We define a novel group action on the input space, and with this symmetry demonstrate that the full second order kinematics are equivariant.
Although the symmetry for the tangent bundle of a Lie-group has been known since the seventies \cite{Brockett72}, we believe this is the first paper where equivariance of the natural second order kinematics has been demonstrated.
The equivariance of the system can be exploited to lift the kinematics onto the symmetry group and from here it is straightforward to build an observer based on constructive nonlinear control techniques.
In this paper, we consider the case where the state can be partially measured.
The performance of the proposed observer is demonstrated on a simulation of a simple hovercraft model equipped with an inertial measurement unit for velocity measurement.

The remainder of the paper consists of seven sections and an appendix.
Section~\ref{sec:prelim} discusses preliminaries and notations that will be used in later part of the paper. 
Section~\ref{sec:problem} presents the problem formulation, where the second order system states, kinematics and output are defined. Section~\ref{sec:symmetry} describes the symmetry group and the equivariant group actions on the states and inputs.
Section~\ref{sec:lift} covers the equivariant lift, origin point and the projected kinematics.
Section~\ref{sec:observer} details the proposed observer on the symmetry group, along with the convergence proof.
Section~\ref{sec:experiment} shows the simulation experiment of a hovercraft moving on a 2D plane, which verifies the performance of our proposed observer.
Finally, the conclusions are drawn in Section~\ref{sec:conclusion}.
The appendix provides a modified version of Barbalat's Lemma \cite{ms93-rap} used in the main convergence proof.

\section{Preliminaries and Notation}
\label{sec:prelim}
Let $\grpG \subset \R^{N \times N}$ be a matrix Lie group and let $\gothg \subset \R^{N \times N}$ be its Lie-algebra as a linear subspace of matrix space.
Let $\pr_\gothg : \R^{N \times N} \to \gothg$ denote the projection of $\R^{N \times N}$ onto $\gothg$ with respect to the Euclidean trace inner product.
That is, for any $\Gamma \in \gothg$ and any $M \in R^{N \times N}$, then
\[
\trace{\Gamma^\top M } = \trace{\Gamma^\top \pr_\gothg (M)}.
\]

Define a Lie-group $\grpG^\ltimes_\gothg = \grpG \ltimes \gothg$ as a semi-direct product of $\grpG$ with $\gothg$, with group multiplication given by
\begin{equation}
(A,a) \cdot (B,b) = (AB, a + \Ad_A b),
\end{equation}
for $(A,a)$ and $(B,b) \in \grpG^\ltimes_\gothg$.
The group identity is $(I_4, 0)$ and group inverse is given by
\begin{equation}
(A,a)^{-1} = (A^{-1}, -\Ad_{A^{-1}} a).
\end{equation}
This group was first proposed by Brockett and Sussmann~\cite{Brockett72}.

The Lie-algebra of $\grpG^\ltimes_\gothg$ is denoted $\gothg^\ltimes_\gothg$ and has elements $(w_1,w_2) \in \gothg \times \gothg$.
We identify the Lie-algebra $\grpG^\ltimes_\gothg = T_{(I,0)} \grpG^\ltimes_\gothg$ with the tangent space of $\grpG^\ltimes_\gothg$ at the identity.

The group $\grpG^\ltimes_\gothg$ is not a matrix group as written and some care must be taken with basic operations such as left and right translation and computing adjoint operators.
In particular, we will need the left differential $\td L_{(A,a)} : T_{(I,0)} \grpG^\ltimes_\gothg \to T_{(A,a)} \grpG^\ltimes_\gothg$ in order to parameterize the tangent space
\[
T_{(A,a)} \grpG^\ltimes_\gothg = \{ \td L_{(A,a)} (w_1,w_2) \;|\; (w_1,w_2) \in \gothg^\ltimes_\gothg \},
\]
and the Adjoint
$\Ad_{(A,a)} : \gothg^\ltimes_\gothg \to \gothg^\ltimes_\gothg$ to verify equivariance properties of the lifted system.

\begin{lemma} \label{lem:dL}
The differential of left translation on $\grpG^\ltimes_\gothg$ at the origin is given by
\begin{align}
\left. \td L_{(A,a)} \right|_{(I,0)} [(w_1, w_2)] = (A w_1, \Ad_A w_2).
\label{eq:dL_A}
\end{align}

The Adjoint is given by
\begin{align}
\label{eq:Adjoint}
\Ad_{(A,a)} (w_1,w_2) = \big( \Ad_A w_1, \Ad_A w_2  - [\Ad_A w_1 , a ] \big),
\end{align}
where $[\cdot, \cdot]$ denotes the Lie-bracket on $\gothg$.
\end{lemma}

\begin{proof}
Computing the differential of the left translation
\begin{align*}
\tD_{(B,b)}\Big|_{(I,0)} &((A,a)\cdot(B,b))[(w_1, w_2)] \\
&= \tD_{(B,b)}\Big|_{(I,0)} (AB, a + \Ad_A b)[(w_1, w_2)] \\
&= (A w_1, \Ad_A w_2).
\end{align*}
This proves \eqref{eq:dL_A}.

In order to demonstrate the result for the Adjoint it is necessary to compute
$\left. \td R_{(A,a)^{-1}} \right|_{(A,a)} [(A w_1, \Ad_A w_2)]$.
One has
\begin{align*}
&\left. \td R_{(A,a)^{-1}} \right|_{(A,a)} [(A w_1, \Ad_A w_2)]
 \\
&\hspace{0.0cm}= \left. \tD_{(B,b)} \left( (B,b) \cdot (A^{-1}, -\Ad_{A^{1}} a) \right) \right|_{(A,a)} [(A w_1, \Ad_A w_2)] \\
&\hspace{0.0cm}= \big( A w_1A^{-1}, \Ad_A w_2 \\
& \hspace{0.4cm}- A w_1 (\Ad_{A^{-1}} a) A^{-1} + A (\Ad_{A^{-1}} a) A^{-1} A w_1 A^{-1} \big) \\
&\hspace{0.0cm}= \left( \Ad_A w_1, \Ad_A w_2 - [ \Ad_A w_1,  a]\right).
\end{align*}
\end{proof}

\section{Problem Formulation}
\label{sec:problem}
The state $\xi$ of second order kinematics of a system on a Lie group $\grpG$ are naturally posed on the tangent bundle $\tT \grpG$.
By using the standard left trivialisation of the velocity state, a tangent vector $P_\xi V_\xi  \in \tT_{P_\xi} \grpG$ is identified with $V_\xi \in \gothg$.
In this manner, the state of our system is posed as an element of the product manifold $\xi = (P_\xi, V_\xi) \in \grpG \times \gothg$ corresponding to the element $ (P_\xi, P_\xi V_\xi) \in \tT \grpG$.
We will abuse notation and identify $\tT\grpG  \equiv \grpG \times \gothg$ in the sequel and write $\xi = (P_\xi, V_\xi)  \in \tT \grpG$.

The first order kinematics of $P_\xi$ are given by $\dot{P}_\xi = P_\xi V_\xi$ for $V_\xi \in \gothg$ the Lie-algebra of $\grpG$.
Since $V_\xi \in \gothg$ lies in a linear vector space, the second order kinematics can be written
\begin{subequations}
\label{eq:classical_kin}
\begin{align}
\dot{P}_\xi &= P_\xi V_\xi, \\
\dot{V}_\xi &= W,
\end{align}
\end{subequations}
where $W \in \tT \gothg \equiv \gothg$ is the generalised acceleration of the system that we will assume is measured.

A key innovation of this paper is to model the full double tangent bundle $\tT (\tT \grpG)$ as input to the second order kinematics.
The double tangent is identified with the tangent of the direct product $\tT (\tT \grpG) \equiv \tT (\grpG \times \gothg)$ using identification discussed above.
In turn we identify $\tT \grpG \times \tT \gothg \equiv \gothg \times \gothg$ with a double copy of the Lie-algebra using left trivialisation in the first part of the direct product and the identification $\tT \gothg \triangleq \gothg$ since $\gothg$ is a linear space.
Thus, inputs are modeled as elements of a vector space
\begin{align}
\vecV = \gothg \times \tT \gothg = \gothg \times \gothg ,
\label{eq:input_space}
\end{align}
where we preserve the $\tT \gothg$ notation to make clear the part of the velocity space that is modelling inputs to the second order part of the kinematics.
An element $U = (U_1 , U_2) \in \vecV$ in the input space can be thought of as two independent elements $U_1 \in \gothg$ and $U_2 \in \tT \gothg \equiv \gothg$.
Note that $V_\xi \in \gothg$ is part of the state $\xi = (P_\xi,V_\xi) \in \tT\grpG $ and not part of the input space.
For the kinematics \eqref{eq:classical_kin} the only measured velocity is $W \in \tT \gothg$ and hence the input velocity that we will use in the observer implementation is $(U_1 , U_2) = (0,W) \in \vecV $.
However, modelling $U_1$ as an independent input, even though eventually it will be set to zero, is critical for the development of the equivariance properties of second order systems.
Thus, the system kinematics we will consider are given by
\begin{subequations}
\label{eq:full_kinematics}
\begin{align}
    \dot{P_\xi} &= P_\xi(V_\xi + U_1), \\
    \dot{V}_\xi &= U_2. \label{eq:dotV}
\end{align}
\end{subequations}

The configuration output that we model is a smooth function $h : \grpG  \rightarrow \calN$ into a smooth manifold $\calN$.
For simplicity, in the present paper, we assume that the full configuration state $P_\xi \in \grpG$ is measured.
That is $h : \grpG \to \grpG$ is the identity map, and
\begin{align}
\label{eq:output}
y = h(P_\xi) = P_\xi.
\end{align}

\begin{remark}
In general, the output function can be used to model measurements from the full state space $h : \tT \grpG \to \calN$ where $\calN$ is a vector space bundle that includes a linear part modeling the velocity measurements.
We postpone a full development of this case to future work.
\end{remark}

We will use the compact notation
\begin{equation} \label{eq:system}
\dot{\xi} = f(\xi, U) = (P_\xi(V_\xi + U_1), U_2).
\end{equation}
for $\xi = (P_\xi, V_\xi)$ and $U = (U_1,U_2)$.
Note that the kinematics \eqref{eq:system} contain a drift term
\[
f((P_\xi, V_\xi), 0) = (P_\xi V_\xi, 0)
\]
for $U \equiv 0$ similar to the drift considered in the authors earlier work on second order attitude kinematics \cite{Ng19}.
This is characteristic of all second order kinematic systems.

\section{Symmetry Of System}
\label{sec:symmetry}
Physical systems usually have physical models with symmetries that encode the equivariance of the laws of motion. When viewed through a symmetric transformation of space, the behaviour of the system at one point is the same as the behaviour at another point in the state space. A system with symmetry allows a global analysis of an observer by analysing the behaviour at one point in space~\cite{Mahony13,2020_mahony_EquivariantSystemsTheory,Ng19}.

\subsection{State space symmetry}
Right multiplication defines a group action of $\grpG$ on $\grpG$ in a natural manner.
Since this action will play a fundamental role as the symmetry operator in our theory, we introduce specific notation $\Phi: \grpG \times \grpG \to \grpG$ by
\begin{align*}
\Phi(A, P) = P A
\end{align*}
for the group action.

\begin{lemma} \label{lem:phi}
Define $\phi: \grpG^\ltimes_\gothg \times \tT\grpG  \to \tT\grpG $ by
\begin{align}
\label{eq:phi_func}
    \phi((A,a),(P_\xi,V_\xi)) = (P_\xi A, \Ad_{A^{-1}} (V_\xi-a)).
\end{align}
The map $\phi$ is a transitive right group action of $\grpG^\ltimes_\gothg$ on $\tT\grpG$.
\end{lemma}

\begin{proof}
Let $(A,a), (B,b) \in \grpG^\ltimes_\gothg$. Then,
\begin{align*}
    \phi((A,a), &\phi((B,b),(P_\xi,V_\xi))) \\
    &= (P_\xi BA, \Ad_{A^{-1}B^{-1}}(V_\xi-b - \Ad_{B}a))\\
    &= (P_\xi(BA), \Ad_{(BA)^{-1}}(V_\xi - (b + \Ad_{B}a)))\\
    &= \phi((B,b)\cdot(A,a),(P_\xi,V_\xi)).
\end{align*}

This demonstrates that the (right handed) group action property holds.
It is straightforward to verify that $\phi((\Id_4, 0), (P_\xi,V_\xi)) = (P_\xi,V_\xi)$.

To see that $\phi$ is transitive, let $(P_\xi,V_\xi)$ and $(P_\xi',V_\xi')$ be any elements of $\tT\grpG $. Then we can find the group element $({P_\xi}^{-1} P_\xi', V_\xi - \Ad_{{P_\xi}^{-1}P_\xi'} V_\xi') \in \grpG^\ltimes_\gothg$ such that
\begin{align*}
    \phi((&{P_\xi}^{-1} P_\xi', V_\xi - \Ad_{{P_\xi}^{-1}P_\xi'} V_\xi'), (P_\xi,V_\xi)) \\
    &= (P_\xi {P_\xi}^{-1} P_\xi', \Ad_{{P_\xi'}^{-1}P_\xi}(V_\xi - (V_\xi - \Ad_{{P_\xi}^{-1}P_\xi'} V_\xi'))) \\
    &= (P_\xi',V_\xi').
\end{align*}
\end{proof}

\subsection{Equivariance of second order kinematics}
In this section we show that the kinematics \eqref{eq:system} are equivariant \cite{Mahony13,2020_mahony_EquivariantSystemsTheory}.
To begin, we need to introduce a group action on the velocity input space.

\begin{lemma} \label{lem:psi}
Define a map $\psi : \grpG^\ltimes_\gothg \times \vecV \to \vecV$ by
\begin{align} \label{eq:psi_func}
    \psi((A,a),(U_1, U_2)) := (\Ad_{A^{-1}}(U_1+a), \Ad_{A^{-1}} U_2).
\end{align}
The map $\psi$  is a right action of $\grpG^\ltimes_\gothg$ on $\vecV$.
\end{lemma}

\begin{proof}
Let $(A,a),(B,b) \in \grpG^\ltimes_\gothg$ and let $(U_1, U_2) \in \vecV$. Then we have
\begin{align*}
    \psi((A,a),&\psi((B,b),(U_1, U_2))) \\
    &= (\Ad_{A^{-1}}(\Ad_{B^{-1}}(U_1+b) + a), \Ad_{A^{-1}}\Ad_{B^{-1}} U_2) \\
    &= (\Ad_{(BA)^{-1}}(U_1+ (b + \Ad_{B} a)), \Ad_{(BA)^{-1}} U_2) \\
    &= \psi((B,b) \cdot (A,a),(U_1, U_2)).
\end{align*}
This shows that the (right handed) group action property holds. It is straightforward to verify that $\psi((\Id_4,0),(U_1, U_2) = (U_1, U_2)$. Thus, $\psi$ is indeed a right action.
\end{proof}

\begin{lemma} \label{lem:equivariant_input}
The system defined in \eqref{eq:system} is equivariant with respect to the group $\grpG^\ltimes_\gothg$ and group actions $\phi$ \eqref{eq:phi_func} and $\psi$ \eqref{eq:psi_func}; that is,
\begin{align*}
    \td \phi_X [f(\xi, U)] = f(\phi_X(\xi), \psi_X(U))
\end{align*}
for any $X \in \grpG^\ltimes_\gothg$, $\xi \in \tT\grpG $, and $v \in \vecV$.
\end{lemma}
\begin{proof}
Let $(A,a) \in \grpG^\ltimes_\gothg$, $(P_\xi, V_\xi) \in \tT\grpG $, and $(U_1, U_2) \in \vecV$ be arbitrary.
Then we have
\begin{align*}
    \td \phi&_{(A,a)}  f((P_\xi, V_\xi), (U_1,U_2)) \\
    &= \tD_{(P_\xi,V_\xi)} \phi_{(A,a)} (P_\xi,V_\xi) [(P_\xi(V_\xi+U_1),U_2)] \\
    &= (P_\xi(V_\xi+U_1)A, \Ad_{A^{-1}}U_2)\\
    &= (P_\xi A(\Ad_{A^{-1}}(V_\xi-a) + \Ad_{A^{-1}}(U_1+a)), \Ad_{A^{-1}} U_2) \\
    &= f( \phi_{(A,a)} (P_\xi,V_\xi), \psi_{(A,a)} (U_1,U_2) ),
\end{align*}
which proves the equivariance condition.
\end{proof}

\begin{remark}
This result cannot be shown without exploiting the velocity input $U_1$.
That is, in order to demonstrate equivariance of second order kinematics, it is necessary to model a first order velocity input separate to the velocity state of the system.
The authors believe that this insight is a key contribution of the paper.
\end{remark}

\section{System Lift}
\label{sec:lift}
\subsection{System Lift onto the Group}
The proposed observer state is posed on the symmetry group \cite{Mahony13,2020_mahony_EquivariantSystemsTheory}.
This allows the state to be defined with respect to an arbitrarily chosen \emph{origin point} while ensuring the same system kinematics is still satisfied.
In order to build the observer it is necessary to find a lift, $\Lambda : \tT \grpG \times \vecV \to \gothg^\ltimes_\gothg$, of the system kinematics into the Lie-algebra $\gothg^\ltimes_\gothg$ of the group $\grpG^\ltimes_\gothg$.
A function $\Lambda : \tT \grpG \times \vecV \to \gothg^\ltimes_\gothg$ is a lift \cite{2020_mahony_EquivariantSystemsTheory} if
\begin{align}
\td\phi_{\xi} [\Lambda(\xi, U)] = f(\xi, U)
\label{eq:lift_condition}
\end{align}
for all $A \in \grpG$, $\xi \in \tT \grpG$ and $U \in \vecV$.

\begin{lemma} \label{lem:system_lift}
Define a function $\Lambda: \tT\grpG  \times \vecV \to \gothg^\ltimes_\gothg$ by
\begin{align}
    \Lambda((P_\xi, V_\xi), (U_1, U_2)) := ( V_\xi + U_1, [ V_\xi, U_1 ] - U_2),
    \label{eq:system_lift}
\end{align}
where the square bracket represents the Lie-bracket on $\gothg$ given by the  matrix commutator for a matrix Lie-group.
The function $\Lambda$ (\ref{eq:system_lift}) is a lift function for \eqref{eq:system} with respect to the symmetry group $\grpG^\ltimes_\gothg$.
\end{lemma}

\begin{proof}
Let $(P_{\xi}, V_{\xi}) \in \tT\grpG $, $(U_1, U_2) \in \vecV$, and $(A,a) \in \grpG^\ltimes_\gothg$ be arbitrary.
Compute
\begin{align*}
    \td&\phi_{(P_\xi, V_\xi)} [\Lambda((P_\xi, V_\xi), (U_1, U_2))] \\
    &= \tD_{(A,a)}\Big|_{I,0} (P_\xi A, \Ad_{A^{-1}}(V_\xi - a))[\Lambda((P_\xi, V_\xi), (U_1, U_2))] \\
    &= (P_\xi (V_\xi + U_1), -\Big[ V_\xi + U_1, V_\xi \Big] - (\Big[ V_\xi, U_1 \Big] - U_2)) \\
    &= (P_\xi (V_\xi + U_1), U_2) \\
    &= f((P_\xi, V_\xi), (U_1, U_2)),
\end{align*}
which verifies \eqref{eq:lift_condition}.
\end{proof}

\begin{lemma} \label{lem:equivariant_lift}
The lift defined in \eqref{eq:system_lift} is equivariant with respect to the group $\grpG^\ltimes_\gothg$ and group actions $\phi$ \eqref{eq:phi_func} and $\psi$ \eqref{eq:psi_func} \cite{2020_mahony_EquivariantSystemsTheory}.  That is
\begin{align*}
    \Ad_{(A,a)} \Lambda(\xi, U) = \Lambda(\phi_{(A,a)^{-1}}(\xi), \psi_{(A,a)^{-1}}(U)),
\end{align*}
for any $(A,a) \in \grpG^\ltimes_\gothg$, $\xi \in \tT\grpG $, and $U \in \vecV$.
\end{lemma}

\begin{proof}
Let $X = (A,a) \in \grpG^\ltimes_\gothg$.
Recalling Lemma~\ref{lem:dL}, direct computation yields
\begin{align*}
& \Ad_{(A,a)} \Lambda(\xi, U) \\
&= \Big(\Ad_A (V_\xi + U_1),  \\
& \quad\quad\quad\quad\quad
\Ad_A (\Big[V_\xi, U_1\Big] - U_2) -\Big[\Ad_A (V_\xi + U_1), a\Big]\Big) \\
&= \Lambda\Big( (P_\xi A^{-1}, \Ad_A(V_\xi+\Ad_{A^{-1}}a)), \\
& \quad\quad\quad\quad\quad
(\Ad_A (U_1 - \Ad_{A^{-1}}a), \Ad_A U_2) \Big) \\
&= \Lambda(\phi_{X^{-1}}(\xi), \psi_{X^{-1}}(U)).
\end{align*}
\end{proof}

\subsection{Origin Point and Projected Kinematics}

Define the lifted system to be
\begin{align}
    \dot{X} := \td \text{L}_X \Lambda(\xi, U), \label{eq:lifted_kin}
\end{align}
for $\xi \in \tT\grpG $, $U \in \vecV$ and $\td \text{L}_X$ given by Lemma~\ref{lem:dL}.
Let $X = (A, a) \in \grpG^\ltimes_\gothg$ be a solution to the lifted system with initial condition $X(0) \in \grpG^\ltimes_\gothg$ such that $\phi_{\mr{\xi}} (X(0)) = \xi(0)$, the true initial condition of the second order kinematics.

Fix an arbitrary origin $\mr{\xi} = (P_{\mr{\xi}}, V_{\mr{\xi}}) \in \tT\grpG$.
This \emph{origin point} defines a global parametrization $\phi_{\mr{\xi}} : \grpG^\ltimes_\gothg \to \tT\grpG $ of the state space $\tT\grpG $ by the symmetry group $\grpG^\ltimes_\gothg$.
From \cite[Lemma 4.9]{2020_mahony_EquivariantSystemsTheory}, the solution $X(t, X(0))$ for $U = (0, W)$ projects back to the state $\xi(t, \mr{\xi})$ via the group action
\begin{align*}
\phi_{\mr{\xi}}(X (t; X(0))) = \xi(t; \mr{\xi}).
\end{align*}
Writing this in explicit coordinates, one has
\begin{align*}
    (P_\xi, V_\xi) = \phi_{(P_{\mr{\xi}}, V_{\mr{\xi}})} (A,a) = (P_{\mr{\xi}} A, \Ad_{A^{-1}}(V_{\mr{\xi}} - a)), \\
    (\hat{P}_\xi, \hat{V}_\xi) = \phi_{({P_{\mr{\xi}}}, {U}_0)}, (\hat{A},\hat{a}) = ({P_{\mr{\xi}}} \hat{A}, \Ad_{\hat{A}^{-1}}({V}_{\mr{\xi}} - \hat{a})).
\end{align*}
Substituting for (\ref{eq:system_lift}) and (\ref{eq:lifted_kin}), the explicit form of the lifted kinematics on $\grpG^\ltimes_\gothg$ can be written as
\begin{subequations}
\label{eq:lifted_system_kinematics}
\begin{align}
    \dot{A} &= A(\Ad_{A^{-1}}(V_{\mr{\xi}} - a) + U_1), \\
    \dot{a} &= \Ad_{A}\left(\Big[\Ad_{A^{-1}}(V_{\mr{\xi}} - a), U_1 \Big] - U_2\right).
\end{align}
\end{subequations}

It is interesting to note that these lifted kinematics do not resemble the system kinematics \eqref{eq:classical_kin}.
If one chooses a reference $V_{\mr{\xi}} = 0$ and sets $U_1 \equiv 0$ and $U_2 = W$ then one recovers $\dot{A} = -aA$ and $\dot{a} = - \Ad_A W$ which at least contains structure similar to the system kinematics.
It is, however, best not to think of the state of lifted system as the physical variables of the system, rather as transformations that relate a reference to the system variables.
This is particularly important when the observer state model is based on the lifted system and not on the actual system, and thus, the observer state itself cannot be directly related to the system evolution, only its projection through the group action to a state estimate can be compared with the system state.

\section{Observer Design}
\label{sec:observer}
Our proposed observer has its state on the symmetry group $\grpG^\ltimes_\gothg$ and uses the lifted system as its internal model.
We use $\hat{X}(t; \hat{X}(0)) \in \grpG^\ltimes_\gothg$ to denote the estimate for the lifted system state $X(t; X(0))$ for unknown $X(0)$.
The observer is based on a \emph{pre-observer} or internal model (a copy of (\ref{eq:lifted_system_kinematics})) with innovation.
The innovation takes output $y$ and the observer state $\hat{X}$, and generates a correction term for the observer dynamics with the goal that $\hat{\xi} = \phi(\hat{X}, \mr{\xi})$ converges to $\xi(t, \mr{\xi})$.

\subsection{Proposed Observer}
\label{th:observer_lifted}
Let $(P_{\mr{\xi}}, V_{\mr{\xi}})$ be the chosen reference point in $\tT\grpG $.
Let $\hat{X} = (\hat{A}, \hat{a}) \in \grpG^\ltimes_\gothg$, with arbitrary initial condition $\hat{X}(0) = (\hat{A}_0, \hat{a}_0)$.
Define  observer kinematics
\begin{subequations}
\label{eq:lifted_observer_kinematics}
\begin{align}
    \dot{\hat{A}} &= \hat{A}(\Ad_{\hat{A}^{-1}}(V_{\mr{\xi}}-\hat{a}) + U_1) - \Delta_1 \hat{A}, \\
    \dot{\hat{a}} &= \Ad_{\hat{A}}\left(\Big[ \Ad_{\hat{A}^{-1}}(V_{\mr{\xi}} - \hat{a}), U_1 \Big] - U_2\right) - \Delta_2,
\end{align}
\end{subequations}
for arbitrary inputs $(U_1,U_2)$ and innovation terms $(\Delta_1, \Delta_2) \in \gothg \times \gothg$ that remain to be defined.

Let $X = (A, a)$ be the lifted state and define the group error to be
\begin{align}
E = \hat{X} X^{-1} = (\tilde{A}, \tilde{a}).
\label{eq:E}
\end{align}
The explicit form of $(\tilde{A}, \tilde{a})$ is
\begin{align}
\tilde{A} = \hat{A} A^{-1}, \quad\quad\quad\quad
\tilde{a} = \hat{a} - \Ad_{\tilde{A}} a .
\label{eq:tilde_Aa}
\end{align}

Given that the initial condition for the lifted system is chosen to match the true system then $V_\xi=\Ad_{A^{-1}}(V_{\mr{\xi}}-a)$.
Define $\hat{V}_\xi=\Ad_{\hat{A}^{-1}}(V_{\mr{\xi}}-\hat{a}$).
The dynamics of the group error $E = (\tilde{A}, \tilde{a})$ are given by
\begin{subequations}
\label{eq:errorsystem_kinematics}
\begin{align}
\dot{\tilde{A}} &= (\Ad_{\hat{A}}\tilde{V}_\xi - \Delta_1 ) \tilde{A}, \\
\dot{\tilde{a}} &= \Ad_{\hat{A}}\left(\Big[ \tilde{V}_\xi, U_1\Big]\right) +\Big[ \Ad_{\tilde{A}} a, \Ad_{\hat{A}}\tilde{V}_\xi - \Delta_1  \Big] - \Delta_2,
\end{align}
\end{subequations}
with $\tilde{V}_\xi = \hat{V}_\xi - {V}_\xi$.

The goal is to design the innovations $\Delta_1$ and $\Delta_2$ such that the group error $E=(\tilde{A}, \tilde{a})=(I,0)$ is locally asymptotically stable to $(I,0)$.
This will in turn ensure asymptotic convergence of $\hat{\xi} = \phi(\hat{X}, \mr{\xi})$ to $\xi(t, \mr{\xi})$.
It is impossible to guarantee the global asymptotic stability of the error without specific structure for the group $\grpG$.
Indeed, many Lie-groups include topological obstructions to existence of global smooth stabilisation controls and in these cases the error system will always have critical points other than the desired equilibria.
In practice, the natural structure of $\grpG^\ltimes_\gothg$ tends to lead to large (almost global) basins of attraction for the proposed observer.
Nevertheless, in the present paper, the appropriate property is that of local asymptotic stability.

\begin{proposition}
Consider the lifted kinematic system \eqref{eq:lifted_system_kinematics} and consider the observer defined by \eqref{eq:lifted_observer_kinematics} along with the following expressions for the innovation terms
\begin{subequations}
\begin{align}
\label{eq:innovations}
\Delta_1 &= -k_1 \mathbb{P}_\gothg\Big( (I - \tilde{A}) \tilde{A}^\top \Big), \\
\Delta_2 &= \Ad_{\hat{A}} \mathbb{P}_\gothg \Big( \Big[ \hat{V}_\xi, \Ad_{\hat{A}^{-1}}\Delta_1 \Big] + k_2 \hat{A}^\top (I - \tilde{A})\tilde{A}^\top \hat{A}^{-\top} \Big),
\end{align}
\end{subequations}
with $k_1$ and $k_2$ two positive gains and where $(\tilde{A}, \tilde{a})$ is given by \eqref{eq:tilde_Aa}.
If $X =(A,a)$ and $X^{-1} = (A^{-1}, -\Ad_A a)$ are bounded then, the equilibrium  $E=(\tilde{A}, \tilde{a})=(I,0)$ is locally asymptotically stable.
\end{proposition}

\begin{proof}
Let us consider the following candidate Lyapunov function
\begin{align} \label{eq:lyapunov_definition}
\Lyap = \frac{1}{2} \tr\Big( (I - \tilde{A})(I - \tilde{A})^\top \Big) + \frac{1}{2 k_2} || \tilde{V}_\xi  ||^2_F .
\end{align}
Differentiating the Lyapunov function, and substituting the proposed innovations $\Delta_1$ and $\Delta_2$ one obtains
\begin{align}
\dot{\Lyap} &= \tr\Big( (-\dot{\tilde{A}}) (I - \tilde{A})^\top \Big) + \frac{1}{k_2} \tr \Big( \tilde{V}_\xi^\top \dot{\tilde{V}}_\xi \Big) \notag\\
&= \tr\Big( \tilde{V}_\xi^\top (-\hat{A}^\top (I - \tilde{A})\tilde{A}^\top \hat{A}^{-\top} + \frac{1}{k_2} \Big[ \Ad_{\hat{A}^{-1}}\Delta_1, \hat{V}_\xi \Big] \notag\\
 &\hspace{0.3cm}+ \frac{1}{k_2} \Ad_{\hat{A}^{-1}} \Delta_2) \Big) + \tr\Big( ((I - \tilde{A}) \tilde{A}^\top)^\top \Delta_1 \Big) \notag \\
&= -k_1 \tr\Big( ((I - \tilde{A}) \tilde{A}^\top)^\top \mathbb{P}_\gothg \Big( ((I - \tilde{A}) \tilde{A}^\top) \Big) \Big) \notag \\
&= -k_1 \Big|\Big| \mathbb{P}_\gothg \Big((I - \tilde{A}) \tilde{A}^\top \Big) \Big|\Big|^2_F . \label{dLyap}
\end{align}
The derivative of the Lyapunov function is negative
semi-definite which in turn implies that $\Lyap$ is bounded.
Since $A$ and $A^{-1}$ are also bounded by assumption, it is straightforward to verify that $\ddot{\Lyap}$ is also bounded.
Applying Barbalat's lemma it follows that $\mathbb{P}_\gothg \Big((I - \tilde{A}) \tilde{A}^\top \Big)$ converges to zero.
It is easily verified that there exists $\ub{\mu} < \ob{\mu}$ such that
\[
\ub{\mu}  \|I-\tilde{A}\|_F^2 \leq \Big|\Big| \mathbb{P}_\gothg \Big((I - \tilde{A}) \tilde{A}^\top \Big) \Big|\Big|^2_F \leq \ob{\mu} \|I-\tilde{A}\|_F^2 ,
\]
in the neighbourhood of the identity.
It follows that $\tilde{A}$ converges to $I$.

It remains to show that $\tilde{a}$ also converges to zero.
The proof relies on a modified version of Barbalat's Lemma shown in the appendix (Lemma~\ref{lem:barbalat}) \cite{ms93-rap}.
From \eqref{eq:errorsystem_kinematics}, the derivative of $\tilde{A}$ can be rewritten as
\begin{equation}
\dot{\tilde{A}} = N(t) + M(t)
\end{equation}
with $N(t)=-\tilde{a}\tilde{A}$ and $M(t)=(-(\Ad_{\tilde{A}}V_{\mr{\xi}}-V_{\mr{\xi}})-\Delta_1)\tilde{A}$.
It is straightforward to verify that $\dot{N}(t)$ is bounded (so that $N(t)$ is uniformly continuous).
Since $\tilde{A}$ converges to $I$ and $\Delta_1$ converges to $0$, one verifies that $M(t)$ converges to zeros and hence, applying Lemma~\ref{lem:barbalat}, $\dot{\tilde{A}}$ converges to $0$ which in turn implies that $\tilde{a}$ also converges to zero. As for the stability of the equilibrium $E=(I,0)$, it is a direct consequence of
relations \eqref{eq:lyapunov_definition} and \eqref{dLyap}.
\end{proof}

\section{Simulation}
\label{sec:experiment}

To demonstrate the proposed observer design, we carried out a simulation of a hovercraft moving in the $\R^2$ plane, and specialised the general observer equations to a specific symmetry group.
The state of the hovercraft is parametrised by $\tT\grpG \triangleq \SE(2) \times \se(2)$, representing the pose and velocity of the vehicle.
The initial pose and velocity, and the linear and angular acceleration of the hovercraft were chosen such that the position of the hovercraft satisfied the defining equation of the Lissajous curve $(\sin(t), \sin(2t)) \in \R^2$ for all time $t$.
Such that
\begin{align*}
U_1 = \begin{pmatrix}
0 & 0 & 0\\
0 & 0 & 0\\
0 & 0 & 0
\end{pmatrix}, \hspace{0.3cm}
U_2 = \begin{pmatrix}
0 & -u_{\text{ang}}(t) & u_{\text{lin}}^x(t)\\
u_{\text{ang}}(t) & 0 & u_{\text{lin}}^y(t)\\
0 & 0 & 0
\end{pmatrix}, 
\end{align*}
where $u_{\text{lin}}$ and $u_{\text{ang}}$ are the corresponding linear and angular acceleration. 

The general observer design proposed in Section \ref{sec:observer} was implemented using the symmetry group $\SE(2)^\ltimes_{\se(2)}$ to demonstrate the design for this particular problem.
The origin pose and velocity of the observer were arbitrarily chosen.  
The gains of the observer were chosen to be $k_1 = 1.0$ and $k_2 = 1.0$.
Both the hovercraft motion and the observer equations were implemented using Euler integration with a time step of $dt = 0.001$s, and were simulated for a period of $15$s.
The trajectories of the positions of the hovercraft and the observer are shown in Figure \ref{fig:tse2_trajectory}(a).
The evolution of the log Lyapunov function, $\log(\Lyap)$ where $\Lyap$ is defined as in \eqref{eq:lyapunov_definition}, is shown in Figure \ref{fig:tse2_trajectory}(b).

\begin{figure}[!htb]
    \centering
    \includegraphics[width=0.43\linewidth]{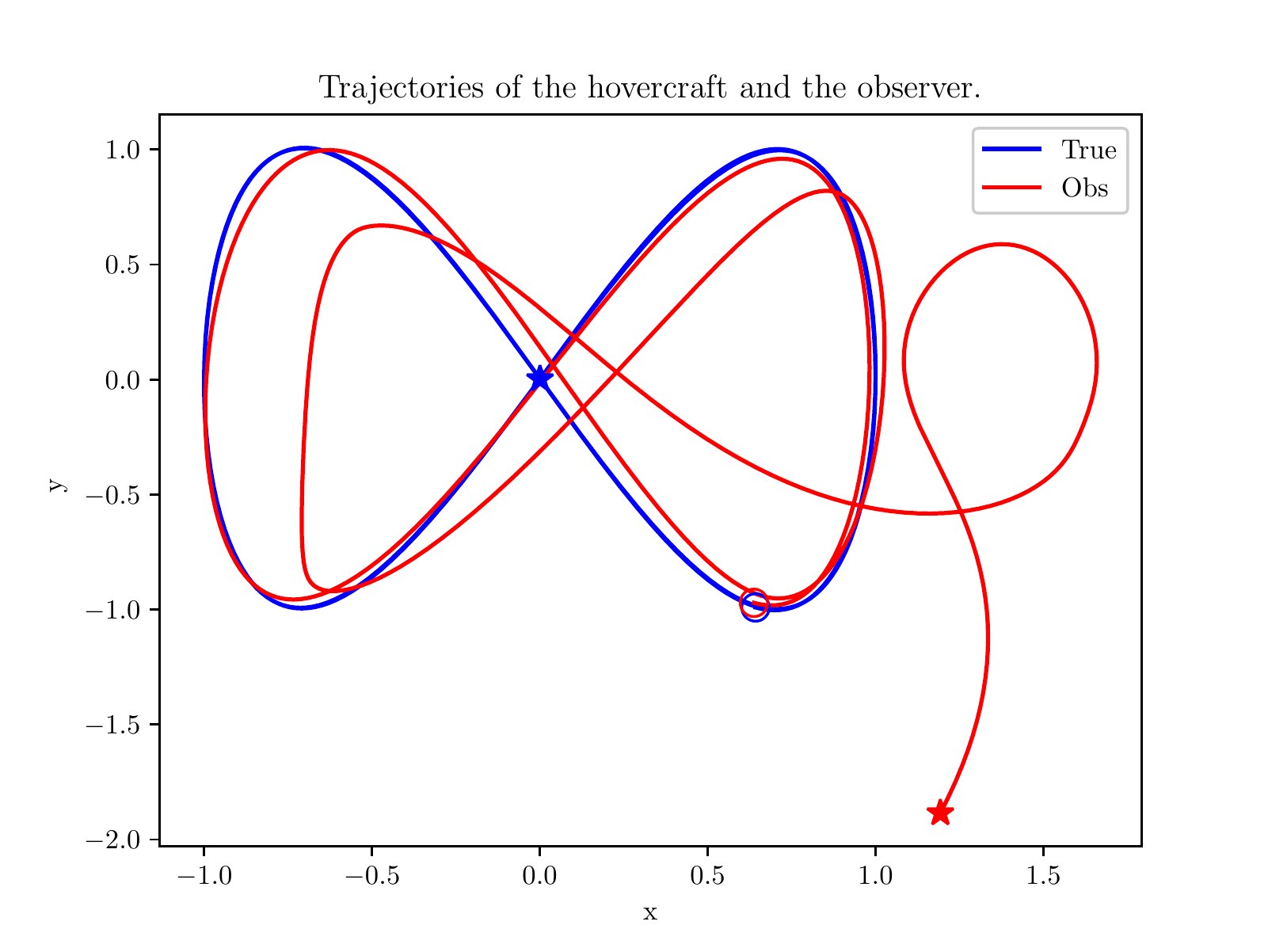}
    \hspace{0.4cm}
    \includegraphics[width=0.43\linewidth]{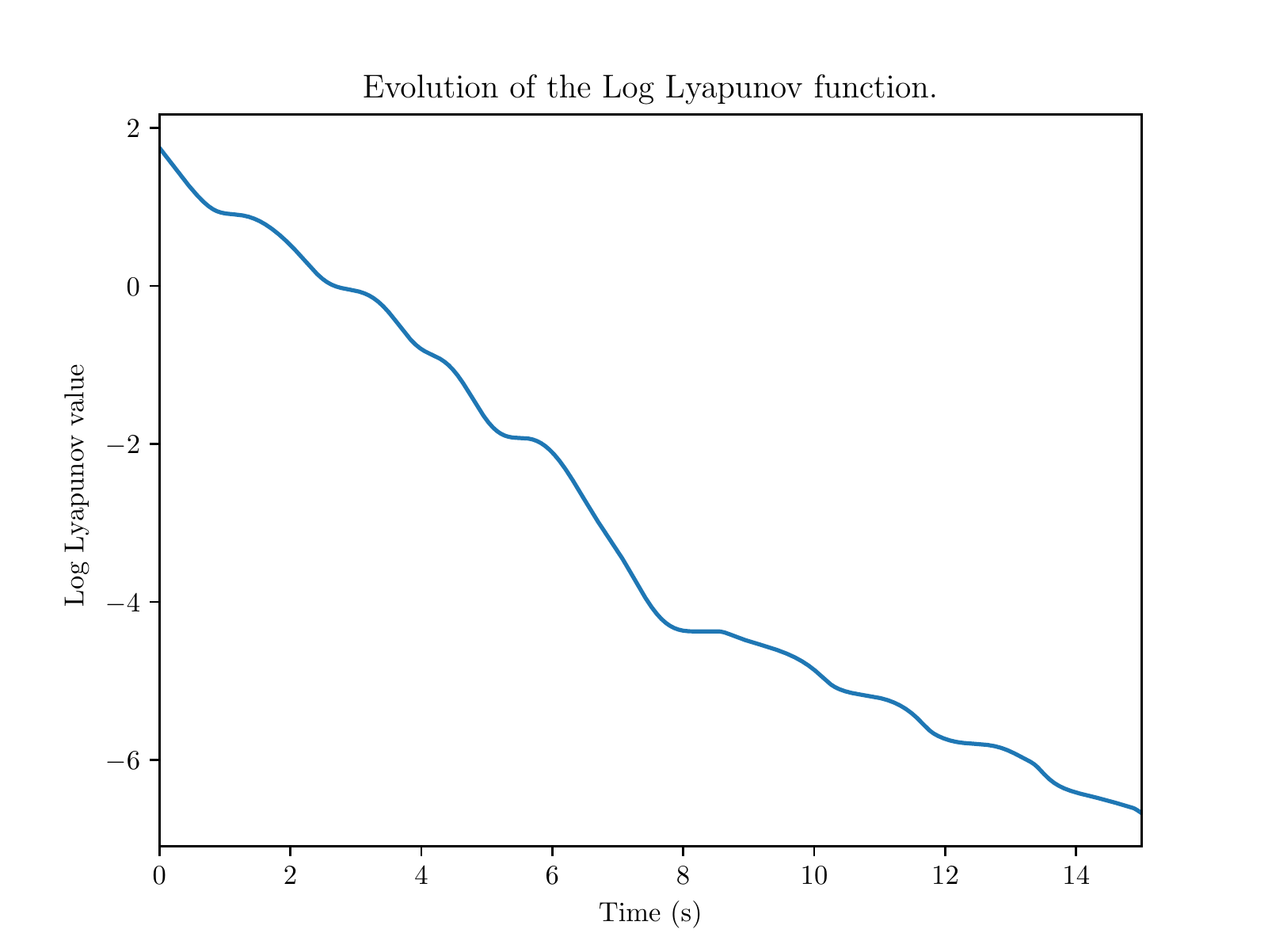}
    \caption{(a) The trajectory of the hovercraft (blue) and its observed trajectory (red).
    A $\star$ and $\circ$ are used to mark the start and end of each trajectory, respectively; (b) The evolution of the Log Lyapunov function showing the local exponential convergence of the observer error system.}
    \label{fig:tse2_trajectory}
\end{figure}

The convergence of the observer trajectory to the hovercraft trajectory shown in Figure \ref{fig:tse2_trajectory}(a) is influenced by the initial error in the velocity estimates.
This causes the shape of the initial trajectory of the observer to differ significantly from that of the hovercraft.
In Figure \ref{fig:tse2_trajectory}(b), the local exponential convergence of the observer is clear from the overall linear decrease of the $\log(\Lyap(t))$ cost.

\section{Conclusions}
\label{sec:conclusion}
In this paper we presented the equivariant systems theory and observer design for second order kinematic systems on matrix Lie group.
More precisely, we identified a symmetry group, its associated equivariant group actions, a system lift, and an observer that operates on the lifted state space.
We also introduce a virtual velocity input that acts on the first order component of the system kinematics, which is vital for the understanding of the equivariance of the system.
The performance of the proposed observer is analysed using Lyapunov stability analysis.
Although global convergence cannot be guaranteed for a general matrix Lie group, local exponential convergence is observed in a simulation of a simple hovercraft model moving on a 2D plane.
This work can be extended to include more general output function $h$ and the estimation of bias in sensor measurements.

\section*{APPENDIX}
\label{app:Barbalat}

\begin{lemma}\label{lem:barbalat} \cite{ms93-rap}
Let $Z(t)$ denote a solution to the differential equation $\dot{Z}(t)=N(t)+M(t)$ with $N(t)$ uniformly continuous function.
Assume that $\lim_{t \rightarrow \infty} Z(t)=C$ and $\lim_{t \rightarrow \infty} M(t)=0$, with $C$ constant.
Then $\lim_{t \rightarrow \infty} \dot{Z}(t)=0$.
\end{lemma}


\bibliographystyle{plain}

\end{document}